\documentclass[10pt]{article}

    \usepackage{amsmath}
    \usepackage{amsthm}
    \usepackage{amssymb}
    \usepackage{graphicx}
    \usepackage{natbib}

    \usepackage[utf8]{inputenc}
    \usepackage[british]{babel}
    \usepackage{mathpartir}
    \usepackage{thmtools}
    \usepackage{thm-restate}
    \usepackage{ifthen}
    \usepackage{caption}
    \usepackage{subcaption}
    \usepackage{url}        

    \usepackage{genMath}
    \usepackage{collections}
    \usepackage{separationLogic}
    \usepackage{obligations}
    \usepackage{notes}
    
    \usepackage{syntax}
    \usepackage{threadPools}
    \usepackage{variables}
    \usepackage{reduction}
    \usepackage{safety}
    \usepackage{proofRules}

\newtheorem{mytheorem}{Theorem}[section]
\newtheorem{mydefinition}[mytheorem]{Definition}

\newtheorem{myobservation}[mytheorem]{Observation}

\title{
    A Separation Logic to Verify Termination of Busy-Waiting for Abrupt Program Exit: \\
    Technical Report
}

\author{
        Tobias Reinhard\\
        KU Leuven
        \and 
        Amin Timany\\
        Aarhus University
        \and 
        Bart Jacobs\\
        KU Leuven
}

\begin{document}

    \NewDocumentCommand{\RRedRuleNameSet}{}{\ensuremath{N_\keywordFont{r}}\xspace}
    \NewDocumentCommand{\progOrdGraph}{}{\ensuremath{\mathcal{G}}\xspace}
    \NewDocumentCommand{\progOrdGraphConfig}{}{\ensuremath{g}\xspace}
    \NewDocumentCommand{\progOrdGraphLF}{}
            {\ensuremath{\progOrdGraph_{\fixedPredNameFont{lf}}}\xspace}
    \NewDocumentCommand{\nodesOf}{m}
            {\ensuremath{\fixedFuncNameFont{nodes}(#1)}\xspace}
    \NewDocumentCommand{\edgesOf}{m}
            {\ensuremath{\fixedFuncNameFont{edges}(#1)}\xspace}
    \NewDocumentCommand{\leavesOf}{m}
            {\ensuremath{\fixedFuncNameFont{leaves}(#1)}\xspace}

    \NewDocumentCommand{\sumOb}{m}{\ensuremath{S_o(#1)}\xspace}
    \NewDocumentCommand{\sumCred}{m}{\ensuremath{S_c(#1)}\xspace}

\maketitle

\begin{abstract}
    Programs for multiprocessor machines commonly perform busy-waiting for synchronisation. 
    In this paper, we make a first step towards proving termination of such programs. 
    We approximate 
    (i)~arbitrary waitable events by abrupt program termination and 
    (ii)~busy-waiting for events by busy-waiting to be abruptly terminated.
    
    We propose a separation logic for modularly verifying termination (under fair scheduling) of programs where some threads eventually abruptly terminate the program, and other threads busy-wait for this to happen.
\end{abstract}

\tableofcontents

\section{Introduction}\label{sec:Introduction}

    Programs for multiprocessor machines commonly perform busy-waiting for synchronisation~\cite{Mhlemann1980MethodFR, MellorCrummey1991AlgorithmsFS, Rahman2012ProcessSI}. 
    In this paper, we make a first step towards proving termination of such programs. 
    Specifically, we propose a separation logic~\cite{Reynolds2002SeparationLA, OHearn2001LocalRA} for modularly verifying termination (under fair scheduling) of programs where some threads eventually abruptly terminate the program, and other threads busy-wait for this to happen.

    Here, by modular we mean that we reason about each thread and each function in isolation.
    That is, we do not reason about thread scheduling or interleavings.
    We only consider these issues when proving the soundness of our logic.

    In this work, we approximate 
    (i)~arbitrary events that a program might wait for by abrupt termination and 
    (ii)~busy-waiting for events by busy-waiting to be abruptly terminated. 
    In Section~\ref{sec:FutureWork}, we sketch preliminary ideas for generalizing this to verifying termination of busy-waiting for arbitrary events, and how this work may also be directly relevant to verifying liveness properties of a program’s I/O behaviour.
    
    Throughout this paper we use a very simple programming language to illustrate our verification approach.
    Its simplicity would allow us to verify termination of busy-waiting for abrupt termination via a static analysis significantly simpler than the proposed separation logic.
    However, in contrast to such an analysis, our approach is also applicable to realistic languages.
    Furthermore, we are confident that the logic we propose can be combined with existing concurrent separation logics like Iris~\cite{Jung2018IrisFT} to verify termination of busy-waiting.

    We start by introducing the programming language in Section~\ref{sec:Language} and continue in Section~\ref{sec:Logic} with presenting the separation logic and so-called \emph{obligations} and \emph{credits}~\cite{Hamin2019TransferringOT, Hamin2018DeadlockFreeM, Leino2010DeadlockFreeCA, Kobayashi2006ANT}, which we use to reason about termination of busy-waiting.
    In Section~\ref{sec:ProofSystem} we present our verification approach in the form of a set of proof rules and illustrate their application.
    Afterwards, we prove the soundness of our proof system in Section~\ref{sec:Soundness}.
    We conclude by outlining our plans for future work, comparing our approach to related work and reflecting on our approach in Sections~\ref{sec:FutureWork}, \ref{sec:RelatedWork} and ~\ref{sec:Conclusion}.

\section{The Language}\label{sec:Language}
    We consider a simple programming language with an \cmdExit command that abruptly terminates all running threads, a \cmdFork command, a looping construct \cmdLoop[\cmdSkip] to express infinite busy-waiting loops and sequencing $\cmdVar_1; \cmdVar_2$.

    \begin{mydefinition}[Commands and Continuations]
        We denote the sets of commands \cmdVar and continuations \contVar as defined by the grammar presented in Figure~\ref{fig:Syntax} by~\CmdSet and~\ContSet.
        We consider sequencing $\cdot\mathop{;}\cdot$ as defined in the grammars of commands and continuations to be right-associative.
    \end{mydefinition}
    
    We use commands and continuations to represent programs and single threads, respectively, as well as natural numbers for thread IDs.
    Continuation \contDone marks the end of a thread's execution.
    We consider thread pools to be functions mapping a finite set of thread IDs to continuations.
        
    \begin{figure}
        $$
        \begin{array}{r c l}
            \cmdVar \in \CmdSet
                    &::=&
                            \cmdExit ~|~
                            \cmdLoop[\cmdSkip] ~|~
                            \cmdFork[\cmdVar] ~|~
                            \cmdVar\ ;\, \cmdVar
            \\
            \contVar \in \ContSet
                    &::=&
                            \contDone   ~|~   \cmdVar\ ;\, \contVar
        \end{array}    
        $$
        \vspace{-0.3cm}
        \caption{Syntax}
        \label{fig:Syntax}
    \end{figure}
    
    \begin{mydefinition}[Thread Pools]
            We define the set of thread pools \ThreadPoolSet as follows
            $$
                \ThreadPoolSet\ := \
                        \{  \tpVar : \tpDomVar \rightarrow \ContSet  
                                \, \ | \, \
                                \tpDomVar \finSubset \N
                        \}.
            $$
            We denote thread pools by \tpVar, thread IDs by \tidVar and the empty thread pool by 
            $\tpEmpty: \emptyset\rightarrow\ContSet$.
        \end{mydefinition}

    \begin{mydefinition}[Thread Pool Extension]\label{def:ThreadPoolExtension}
        Let $\tpVar:\tpDomVar\rightarrow\ContSet \in \ThreadPoolSet$ be a thread pool. We define:
        \begin{itemize}
            \item $\tpVar \tpExt \emptyset := \tpVar$,
            \item 
                    $\tpVar \tpExt \setOf{\contVar} :
                        \tpDomVar\,\cup\,\setOf{\max(\tpDomVar)+1}\rightarrow\ContSet$ 
                    with\\ 
                    $(\tpVar\tpExt\setOf{\contVar})(\tidVar) = P(\tidVar)$ for all $\tidVar\in\tpDomVar$ and \\
             $(\tpVar\tpExt\setOf{\contVar})(\max(\tpDomVar)+1) = \contVar$ ,
            \item 
                    $\tpVar \tpRem \tidVar' : 
                        \tpDomVar \setminus \setOf{\tidVar'}\rightarrow \ContSet$
                    with\\
                    $(\tpVar \tpRem \tidVar')(\tidVar) = \tpVar(\tidVar)$
        \end{itemize}
    \end{mydefinition}
    
    We consider a standard small-step operational semantics for our language defined in terms of two reduction relations: (i)~\stRedStepSymb for single thread reduction steps and (ii)~\tpRedStepSymb for thread pool reduction steps.
    
    \begin{mydefinition}[Single-Thread Reduction Relation]
        We define a \emph{single-thread reduction relation} \stRedStepSymb according to the rules presented in Figure~\ref{fig:SingleThreadReductionRelation}.
        A reduction step has the form
        $$\stRedStep{\contVar}{\nextContVar}[\ftsVar]$$
        for a set of forked threads $\ftsVar \subset \ContSet$ with $\cardinalityOf{\ftsVar} \leq 1$.
    \end{mydefinition}

    \begin{figure}
        \begin{mathpar}
            \RedSTLoop
            \and
            \RedSTFork
            \and
            \RedSTSeq
        \end{mathpar}
        \caption{Reduction rules for single threads.}
        \label{fig:SingleThreadReductionRelation}
    \end{figure}

    \begin{mydefinition}[Thread Pool Reduction Relation]
        We define a \emph{thread pool reduction relation} \tpRedStepSymb according to the rules presented in Figure~\ref{fig:ThreadPoolReductionRelation}. 
        A reduction step has the form
        $$\tpRedStep{\tpVar}{\tidVar}{\nextTpVar}$$
        for a thread ID $\tidVar \in \domOf{\tpVar}$.
    \end{mydefinition}

    \begin{figure}
        \begin{mathpar}
            \RedTPLift
            \and
            \RedTPExit
            \and
            \RedTPThreadTerm
        \end{mathpar}
        \caption{Reduction rules for thread pools.}
        \label{fig:ThreadPoolReductionRelation}
    \end{figure}

    \paragraph{Termination Terminology}
    According to reduction rule \RedTPExitName, command \cmdExit terminates all running threads by clearing the entire thread pool.
    We call this \emph{abrupt termination} to differentiate it from \emph{normal termination} when a thread first reduces to \contDone and is then removed from the thread pool via rule \RedTPThreadTermName.
    The term \emph{termination} encompasses both \emph{abrupt} and \emph{normal} termination.

    Figure~\ref{fig:ExampleProgramCode} illustrates the type of programs we aim to verify.
    The code snippet spawns a new thread which will abruptly terminate the entire program and then busy-waits for the program to be terminated.
    The operational semantics defined above is non-deterministic in regard to when and if threads are scheduled.
    Meanwhile, the presented program only terminates if the exiting thread is eventually scheduled.
    Hence, we need to assume fair scheduling.

    \begin{figure}
            $$
                    \cmdFork[\cmdExit];
                    \cmdLoop[\cmdSkip]
            $$
            \caption{Example program with two threads: An exiting thread and one waiting for the program to be abruptly terminated.}
            \label{fig:ExampleProgramCode}
    \end{figure}

    \begin{mydefinition}[Reduction Sequence]\label{def:ReductionSequence}
        Let  \tpRSeq{\N} be a sequence of thread pools such that \tpRedStep{\tpVar_i}{\tidVar_i}{\tpVar_{i+1}} holds for all $i\in \N$ and some sequence $(\tidVar_i)_{i\in \N}$ of thread IDs. Then we call \tpRSeq{\N} a \emph{reduction sequence}.
    \end{mydefinition}

    Note that according to this definition, all reduction sequences are implicitly infinite.

    \begin{mydefinition}[Fairness]\label{def:FairInfiniteReduction}
        We call a reduction sequence \phantom{} \tpRSeq{\N} \emph{fair} iff for all $k \in \N$ and $\tidVar_k\in\domOf{\tpVar_k}$ 
         there exists $j\geq k$ such that 
        $$
            \tpRedStep{\tpVar_j}{\tidVar_k}{\tpVar_{j+1}}.
        $$
    \end{mydefinition}

\section{The Logic}\label{sec:Logic}
    
    In this paper, we develop a separation logic to reason about termination of busy-waiting programs.
    Separation logic is designed for reasoning about program resources as well as ghost resources~\cite{Reynolds2002SeparationLA, OHearn2001LocalRA}.
    The latter is information attached to program executions for the purpose of program verification, e.g., a resource tracking how many threads have access to a shared memory location~\cite{Jung2016HigherorderGS}.
    Here, we use ghost resources to track which thread will eventually \cmdExit, i.e., abruptly terminate the entire program.
    
    \paragraph{Obligations \& Credits}
    Remember that \cmdExit terminates all running threads.
    Therefore, in order to modularly reason about program termination we need information about other threads performing \cmdExit.
    
    For this purpose, we introduce two kinds of ghost resources: \emph{obligations} and \emph{credits}.
    Threads holding an obligation are required to discharge it by performing \cmdExit while threads holding a credit are allowed to busy-wait for another thread to \cmdExit.
    As seen in the next section we ensure that no thread (directly or indirectly) waits for itself.

    We aggregate obligations into obligations chunks, where each obligations chunk collects the held obligations of a single thread.

    \paragraph{Assertions}
    The language of assertions defined in the following allows us to express knowledge and assumptions about held obligations and credits.
    The language contains the standard separating conjunction $\cdot \slStar \cdot$ as well as two non-standard predicates \obsPred and \credit to express the possession of ghost resources.
    (i)~\obs{n} expresses the possession of one obligations chunk containing $n$ \cmdExit obligations; i.e., it expresses that the current thread holds $n$ exit obligations~\footnote{
        As outlined in Section~\ref{sec:FutureWork}, we plan to extend this logic to one where threads are obliged to set ghost signals.
        This makes it necessary to track the number of signals that remain to be set.
        Hence, we track the number of obligations.
    }.
    (ii)~\credit expresses the possession of an \cmdExit credit that can be used to busy-wait for another thread to \cmdExit.
    
    \begin{mydefinition}[Assertions]
        Figure~\ref{fig:AssertionsSyntax} defines the set of assertions \AssertionSet.
    \end{mydefinition}
    
    \begin{figure}
        $$
        \begin{array}{r c l} 
            \assVar \in \AssertionSet &::= 
                    &\slTrue ~|~ \slFalse ~|~ \assVar \slStar \assVar  ~|~ 
                    \obs{\natVar} ~|~ \credit\\
            \natVar \in \N
        \end{array}
        $$
        \vspace{-0.5cm}
        \caption{Syntax of assertions.}
        \label{fig:AssertionsSyntax}
        \vspace{-0.04cm}
    \end{figure}

    As we see in Section~\ref{sec:ProofSystem} it is crucial to our verification approach that the \obsPred-predicate captures a full obligations chunk and that this chunk can only be split when obligations are passed to a newly forked thread.
    We represent the information about the held obligations chunks and credits by resource bundles $(\obBagVar, \credVar)$.
    
    \begin{mydefinition}[Resource Bundles]
        We define the set of \emph{resource bundles} \ResourceBundleSet as
        $$\ResourceBundleSet\ \ :=\ \ \BagsOf{\N} \times \N.$$
        Let
        $(\obBagVar_1, \credVar_1), (\obBagVar_2, \credVar_2) \in \ResourceBundleSet$.
        We define
        $$
                \begin{array}{l c l}
                        (\obBagVar_1, \credVar_1)\ \tupleCup\
                        (\obBagVar_2, \credVar_2)
                        &:=
                        &(\obBagVar_1 \msCup \obBagVar_2, 
                            \credVar_1 + \credVar_2).
                 \end{array}
        $$
    \end{mydefinition}

    Threads hold exactly one obligations chunk, i.e., resources $(\obBagVar, \credVar)$ with $\cardinalityOf{\obBagVar} = 1$.
    We call such resource bundles \emph{\completeTerm}.
    
    \begin{mydefinition}[Complete Resource Bundles]
        We call a resource bundle $(\obBagVar, \credVar) \in \ResourceBundleSet$ \emph{\completeTerm} if $\cardinalityOf{\obBagVar} = 1$ holds and write \complete{(\obBagVar, \credVar)}.
    \end{mydefinition}
    
    Note that the following definition indeed ensures that the \obsPred-predicate captures a full obligations-chunk.
    Hence, no bundle with one obligations chunk can satisfy an assertion of the form $\obs{n} \slStar \obs{n'}$.

    \begin{mydefinition}[Assertion Model Relation]
        Figure~\ref{fig:AssertionModelRelation} defines the \emph{assertion model relation} $\assModelsSymb\ \subseteq\ \ResourceBundleSet\, \times\, \AssertionSet$.
        We write
        $$ \assModels{\rbVar}{\assVar}$$
        to express that resource bundle $\rbVar\in \ResourceBundleSet$ models assertion $\assVar \in \AssertionSet$.
    \end{mydefinition}

    \begin{figure}

        $$
        \begin{array}{r c l c l}
            \rbVar &\assModelsSymb &\slTrue
            \\
            \rbVar &\assModelsSymb &\assVar_1 \slStar \assVar_2 
                    &\text{iff}
                    &\exists \rbVar_1, \rbVar_2\in \ResourceBundleSet.\ 
                            \rbVar = \rbVar_1\tupleCup \rbVar_2\\
                    &&&&
                            \wedge\
                            \rbVar_1 \assModelsSymb \assVar_1\ \wedge\ 
                            \rbVar_2 \assModelsSymb \assVar_2
            \\
            (\obBagVar, \credVar) &\assModelsSymb &\obs{\obVar}
                    &\text{iff}
                    &\obVar \in \obBagVar
            \\
            (\obBagVar, \credVar) &\assModelsSymb &\credit
                    &\text{iff}
                    &\credVar \geq 1
        \end{array}
        $$
        \caption{Modeling relation for assertions.}
        \label{fig:AssertionModelRelation}
    \end{figure}

\section{Verifying Termination of Busy-Waiting}\label{sec:ProofSystem}
    In this section we present the proof system we propose for verifying termination of programs with busy-waiting for abrupt program exit and illustrate its application.
   Further, we present a soundness theorem stating that every program, which provably discharges all its \cmdExit obligations and starts without credits, terminates.

    \paragraph{Hoare Triples}
    We use Hoare triples \hoareTriple{\htPreConVar}{\cmdVar}{\htPostConVar}~\cite{Hoare1968HoareLogic} to specify the behaviour of programs.
    Such a triple expresses that given precondition \htPreConVar, command \cmdVar can be reduced without getting stuck and if this reduction terminates, then postcondition \htPostConVar holds afterwards.
    In particular, a triple \hoareTriple{\htPreConVar}{\cmdVar}{\slFalse} expresses that \cmdVar diverges or exits abruptly.

    \paragraph{Ghost Steps}
    When verifying the termination of a program \cmdVar, we consider it to start without any obligations or credits, i.e., \hoareTriple{\noObs}{\cmdVar}{\htPostConVar}.
    Obligation-credit pairs can, however, be generated during so-called \emph{ghost steps}.
    These are steps that exclusively exist on the verification level and only affect ghost resources, but not the program's behaviour~\cite{Jung2018IrisFT, Fillitre2016TheSO}.
    A credit can also be cancelled against an obligation.
    
    \paragraph{View Shift}
    In our proofs, we need to capture ghost steps as well as drawing conclusions from assertions, e.g., rewriting $\htPreConVar \slStar \htPostConVar$ into $\htPostConVar \slStar \htPreConVar$ and concluding \noObs from the assumption \slFalse.
    We ensure this by introducing a view shift relation~\viewShiftSymb~\cite{Jung2018IrisFT}.
    A view shift \viewShift{\htPreConVar}{\htPostConVar} expresses that whenever \htPreConVar holds, then either 
    (i)~\htPostConVar also holds or
    (ii)~\htPostConVar can be established by performing ghost steps.
    \lrViewShift{\htPreConVar}{\htPostConVar} stands for
    $\viewShift{\htPreConVar}{\htPostConVar}\, \wedge\,
    \viewShift{\htPostConVar}{\htPreConVar}$.

    \begin{mydefinition}[View Shift]\label{def:ViewShiftRelation}
        We define the view shift relation $\viewShiftSymb\ \subset \AssertionSet \times \AssertionSet$ according to the rules presented in Figure~\ref{fig:ViewShiftRelation}.
    \end{mydefinition}
        
    \begin{figure}
            \begin{mathpar}
                \VSObCredIntro
                \and
                \VSSemImp
                \and
                \VSTrans
            \end{mathpar}
            
            \caption{View shift rules.}
            \label{fig:ViewShiftRelation}
    \end{figure}

    Note that view shifts only allow to spawn or remove obligations and credits simultaneously.
    This way, we ensure that the number of obligations and credits in the system remains equal at any time (provided this also holds for the program's initial state).

    \paragraph{Proof Rules}
    We verify program specifications \hoareTriple{\htPreConVar}{\cmdVar}{\htPostConVar}
    via a proof relation \htProvesSymb defined by a set of proof rules.
    These rules are designed to prove that every command \cmdVar, which provably discharges its obligations, i.e.,
    \htProves{\obs{n}}{\cmdVar}{\noObs},
    terminates under fair scheduling.

    \begin{mydefinition}[Proof Relation]\label{def:ProofRelation}
        We define a proof relation \htProvesSymb for Hoare triples \hoareTriple{\htPreConVar}{\cmdVar}{\htPostConVar}
        according to the rules presented in Figure~\ref{fig:proofRules}.
    \end{mydefinition}

    \begin{figure}
        \begin{mathpar}
            \PRFrame
            \and
            \PRExit
            \and
            \PRLoop
            \and
            \PRFork
            \and
            \PRSeq
            \and
            \PRViewShift
        \end{mathpar}
    
        \caption{Proof rules.}
        \label{fig:proofRules}
    \end{figure}

    Obligation-credit pairs can be generated and removed via a ghost step by applying \PRViewShiftName plus \VSObCredIntroName.
    The only way to discharge an obligation, i.e., removing it without simultaneously removing a credit, is via rule \PRExitName.
    That is, a discharging program \htProves{\obs{1}}{\cmdVar}{\noObs}, must involve an abrupt \cmdExit at some point.
    
    We can pass obligations and credits to newly forked threads by applying \PRForkName.
    However, note that in order to prove anything about a command \cmdFork[\cmdVar], we need to prove that the forked thread discharges or cancels all of its obligations.

    The only way to justify busy-waiting is via \PRLoopName, which requires the possession of a credit.
    Note that the rule forbids the looping thread to hold any obligations.
    This ensures that threads do not busy-wait for themselves to \cmdExit.

    \paragraph{Example}
    Consider the program
    $\cmdExVar = 
        \cmdFork[\cmdExit];
        \cmdLoop[\cmdSkip]
    $
    presented in Figure~\ref{fig:ExampleProgramCode}.
    It forks a new thread instructed to \cmdExit and busy-waits for it to do so.
    We can verify its termination under fair scheduling by proving 
    \htProves{\noObs}{\cmdExVar}{\noObs}.
    Figure~\ref{fig:VerificationExampleCode} sketches this proof.
    Note that the assumption of fair scheduling is essential, since otherwise we would have no guarantees that the exiting thread is ever executed.

    \begin{figure}
            $$
            \begin{array}{l l}
                    \progProof{\obs{0}}\\
                    \progProof{\obs{1} \slStar \credit}
                            &\proofRuleHint{\PRViewShiftName + \VSObCredIntroName}\\
                    \cmdFork
                            &\proofRuleHint{\PRForkName}\\
                            
                            \quad\progProof{\obs{1}}\\
                            \quad\cmdExit;
                                    &\proofRuleHint{\PRExitName}\\
                            \quad\progProof{\slFalse}\\
                            \quad\progProof{\obs{0}}
                                    &\proofRuleHint{\PRViewShiftName + \VSSemImpName}\\
                    
                    \progProof{\obs{0} \slStar \credit}\\
                    \cmdLoop[\cmdSkip]
                            &\proofRuleHint{\PRLoopName}\\
                            
                    \progProof{\slFalse}\\
                    \progProof{\obs{0}}
                            &\proofRuleHint{\PRViewShiftName + \VSSemImpName}
            \end{array}
            $$
            
            \caption
            {
                    Verification sketch for a program with two threads: 
                    An exiting thread and one busy-waiting for abrupt termination.
                    Applied proof rules are highlighted in violet.
            }
            
            \label{fig:VerificationExampleCode}
    \end{figure}

    The following \hspace{0.02cm}soundness\hspace{0.02cm} theorem\hspace{0.02cm} states that we can prove termination of a program \cmdVar under fair scheduling by proving that it discharging all its \cmdExit obligations,
    i.e.,
    \htProves{\obs{n}}{\cmdVar}{\noObs}.
    By such a proof we verify that no fair infinite reduction sequence of \cmdVar exists.
    That is, the reduction eventually terminates, either abruptly via \cmdExit or normally.

    \begin{restatable}[Soundness]{mytheorem}{SoundnessTheorem}
    \label{theo:Soundness}
        Let
        $\htProves{\obs{n}}{\cmdVar}{\noObs}$.
        There exists no fair reduction sequence $(\tpVar_i)_{i\in\N}$ starting with $\tpVar_0 = \setOf{(\tidVar_0, \cmdVar;\contDone)}$ for any $\tidVar_0 \in \N$.
    \end{restatable}

\section{Soundness}\label{sec:Soundness}
    Proving Soundness Theorem~\ref{theo:Soundness} requires us to establish a connection between our proof relation~\htProvesSymb and the operational semantics, i.e., the thread pool reduction relation~\tpRedStepSymb.
    
    \paragraph{Bridging the Gap}
    According to our proof rules, ghost resources are not static but affected by the commands occurring in a program.
    For instance, forking allows us to pass resources to the newly forked thread and exiting discharges obligations (cf. proof rules \PRForkName and \PRExitName+\PRViewShiftName in Figures~\ref{fig:proofRules} and~\ref{fig:ViewShiftRelation}).
    We capture the connection between ghost resources and program executions by annotating threads with resource bundles and introducing an annotated operational semantics mimicking the resource manipulation apparent in our proof rules.

    In our proof rules we used Hoare triples 
    \hoareTriple{\htPreConVar}{\cmdVar}{\htPostConVar}
    to specify the behaviour of a program \cmdVar.
    We interpret such triples in a model relation \htModelsSymb, which we define in terms of the annotated semantics.
    Intuitively, \htModels{\htPreConVar}{\cmdVar}{\htPostConVar}
    expresses that given any resource bundle fulfilling precondition \htPreConVar, we can reduce command \cmdVar in the annotated semantics without getting stuck.
    In case the reduction terminates normally (i.e., in case it neither diverges nor exits abruptly), postcondition \htPostConVar holds afterwards.
    Note that this interpretation complies with the intuition behind our proof rules.
    
    We prove our proof relation \htProvesSymb sound with respect to our model relation \htModelsSymb, i.e., we show that
    \htProves{\htPreConVar}{\cmdVar}{\htPostConVar}
    implies
    \htModels{\htPreConVar}{\cmdVar}{\htPostConVar},
    and establish a connection between the annotated and the plain semantics.
    This establishes the missing link between our proof rules and program executions and thereby allows us to prove the soundness theorem.

\subsection{Annotated Executions}\label{subsec:app:Annotations}
    We use ghost resources to track which threads are obliged to \cmdExit, and which are allowed to busy-wait for another thread to do so.
    In order to associate threads with their respective ghost resources, we introduce an annotated version of thread pools.

    \begin{mydefinition}[Annotated Thread Pools]
        We define the set of \emph{annotated thread pools} \RAThreadPoolSet as follows
        $$
        \begin{array}{r c l}
            \RAThreadPoolSet &:=&
                    \{
                            \rtpVar : \tpDomVar \rightarrow \ResourceBundleSet \times \ContSet
                            \ | \
                            \tpDomVar \finSubset \N
                    \}.
        \end{array}
        $$
        We denote annotated thread pools by \rtpVar and the \emph{empty annotated thread pool} by $\rtpEmpty: \emptyset\rightarrow \ResourceBundleSet \times \ContSet$.
        Furthermore, we define an extension operation \rtpExt and a removal operation \rtpRem analogously to \tpExt and \tpRem, respectively, cf.\@ Definition~\ref{def:ThreadPoolExtension}.
    \end{mydefinition}
    
    As it becomes apparent in the proof rules, a thread's resources are not static but affected by the thread's actions.
    For instance, threads can spawn obligation-credit pairs and pass some of their held resources to newly forked threads 
    (cf. proof rules \PRViewShiftName \& \VSObCredIntroName, \PRForkName in Figures~\ref{fig:proofRules} and~\ref{fig:ViewShiftRelation}).
    To make this precise, we define annotated versions \rstRedStepSymb and \rtpRedStepSymb of the reduction relations \stRedStepSymb and \tpRedStepSymb defined in Section~\ref{sec:Language}.

    \begin{mydefinition}
        Let
        $(\multiset{\obVar_1}, \credVar_1), (\multiset{\obVar_2}, \credVar_2) \in \ResourceBundleSet$.
        We define
        $$
                \begin{array}{l c l}
                         (\multiset{\obVar_1}, \credVar_1)\ \rbPlus\
                         (\multiset{\obVar_2}, \credVar_2)
                         &:=
                         &(\multiset{\obVar_1 + \obVar_2}, 
                              \credVar_1 + \credVar_2).
                 \end{array}
        $$
    \end{mydefinition}
    
    \begin{mydefinition}[Annotated Single-Thread Reduction Relation]
            We define an \emph{annotated single-thread reduction relation} \rstRedStepSymb according to the rules presented in Figure~\ref{fig:ResourceAnnotatedSingleThreadReductionRules}.
            A reduction step has the form
            $$
                \rstRedStep
                        {\rbVar}{\contVar}
                        {\nextRbVar}{\nextContVar}[\rftsVar]
            $$
            for a set of forked annotated threads 
            $\rftsVar \subset \ResourceBundleSet \times \ContSet$ 
            with $\cardinalityOf{\rftsVar} \leq 1$.
        \end{mydefinition}

    \begin{figure}
    \begin{mathpar}

    \end{mathpar}
        \begin{mathpar}
            \RARedSTLoop
            \and
            \RARedSTFork
            \and
            \RARedSTSeq
        \end{mathpar}
        
        \caption{Annotated reduction rules for single threads.}
        \label{fig:ResourceAnnotatedSingleThreadReductionRules}
    \end{figure}
    
    In contrast to the plain semantics, reductions in the annotated semantics can get stuck.
    Note that according to \RARedSTLoopName, performing a loop iteration requires holding an empty obligations chunk.
    This ensures that busy-waiting threads do not wait for themselves and corresponds to the restriction that proof rule \PRLoopName imposes on looping threads.
    Consider the program \cmdLoop[\cmdSkip] busy-waiting for itself.
    Our proof rules do not allow us to prove any specification
    \hoareTriple{\htPreConVar}{\cmdLoop[\cmdSkip]}{\htPostConVar}
    where precondition \htPreConVar is neither \slFalse nor contains any \credit
    and its reduction gets stuck in the annotated semantics.
    
    For the annotated semantics we introduce ghost steps that do not correspond to steps in the unannotated semantics, but only affect the ghost resources we added for verification purposes. 
    In particular, ghost steps allow a thread to spawn an obligation-credit pair
    and to cancel an obligation against a credit.
    We introduce two auxiliary step relations \grtpRedStepSymb and \ngrtpRedStepSymb to clearly differentiate between ghost steps and steps corresponding to \emph{real} program steps.

    \begin{mydefinition}[Ghost Thread Pool Steps]
            We define a \emph{ghost step relation} \grtpRedStepSymb on annotated thread pools according to the rules presented in Figure~\ref{fig:ThreadPoolGhostReductionRules}.
            A ghost step has the form
            $$\grtpRedStep{\rtpVar}{\tidVar}{\nextRtpVar}$$
            for an ID $\tidVar \in \domOf{\rtpVar}$ and only affects the resources associated with \tidVar.
            We denote its reflexive transitive closure by \grtpMultRedStepSymb[\tidVar].
    \end{mydefinition}
    
    \begin{figure}
            \begin{mathpar}
                    \GSObCredIntro
                    \and
                    \GSObCredCancel
            \end{mathpar}
            
            \caption{Ghost step rules for thread pools. Ghost steps allow to spawn and cancel an obligation-credit pair.}
            \label{fig:ThreadPoolGhostReductionRules}
    \end{figure}
    
    Ghost steps reflect the resource manipulation expressed by view shifts.
    A ghost step performed by \GSObCredIntroName spawns an obligation-credit pair while \GSObCredCancelName cancels an obligation against a credit.
    This mimics the treatment of obligations and credits displayed by view shift rule
    \VSObCredIntroName.

    \begin{mydefinition}[Real Thread Pool Reduction Steps]
        We define a \emph{non-ghost reduction relation} \ngrtpRedStepSymb for annotated thread pools according to the rules presented in Figure~\ref{fig:ThreadPoolRealReductionRules}.
        A reduction step has the form
        $$
            \ngrtpRedStep{\rtpVar}{\tidVar}{\nextRtpVar}
        $$
        for a thread ID $\tidVar \in \domOf{\rtpVar}$.
    \end{mydefinition}

    \begin{figure}
        \begin{mathpar}
            \RARedTPLift
            \and
            \RARedTPExit
            \and
            \RARedTPThreadTerm
        \end{mathpar}
        
        \caption{Annotated reduction rules for non-ghost steps of thread pools.}
        \label{fig:ThreadPoolRealReductionRules}
    \end{figure}

    We only allow annotated threads to terminate, 
    i.e., be removed from the thread pool,
    if they do not hold any obligations, as shown by rule \RARedTPThreadTermName.
    
    We define the annotated thread pool reduction relation \rtpRedStepSymb as the union of~\ngrtpRedStepSymb and~\grtpRedStepSymb.
    \begin{mydefinition}[Annotated Thread Pool Reduction Relation]
        We define an \emph{annotated thread pool reduction relation} \rtpRedStepSymb such that:
        $$
                \rtpRedStep{\rtpVar}{\tidVar}{\nextRtpVar}
                \ \ \Longleftrightarrow\ \ 
                \ngrtpRedStep{\rtpVar}{\tidVar}{\nextRtpVar}
                \ \vee\ 
                \grtpRedStep{\rtpVar}{\tidVar}{\nextRtpVar}
        $$
    \end{mydefinition}
    
    Note that our reduction relation \rtpRedStepSymb ensures that at any time the number of spawned credits in the system equals the number of spawned obligations.
    Also, the only way to discharge an obligation, i.e., removing it without simultaneously removing a credit, is by exiting.
    These two properties are crucial to the Soundness proof presented in Section~\ref{subsec:SoundnessProof}.

    \begin{mydefinition}[Annotated Reduction Sequence]
            We define annotated reduction sequences analogously to Definition~\ref{def:ReductionSequence}.
    \end{mydefinition}

    Since our goal is to prove that no fair reduction sequence can correspond to a program discharging all its obligations, we need to lift our fairness definition to the annotated semantics.

    \begin{restatable}[Fair Annotated Reduction Sequences]{mydefinition}{defFairAnnotatedRedSeq}
            We call an annotated reduction sequence \rtpRSeq{\N} fair iff
            for all $k\in \N$ and $\tidVar_k \in \domOf{\rtpVar_k}$
            there exists $j \geq k$ such that
            $$\ngrtpRedStep{\rtpVar_j}{\tidVar_k}{\rtpVar_{j+1}}$$
    \end{restatable}
    
    Note that fairness prohibits threads to perform ghost steps forever.

\subsection{Interpreting Specifications}

    \paragraph{Specifications}
    We use Hoare triples
    \hoareTriple{\htPreConVar}{\cmdVar}{\htPostConVar}
    to express specifications.
    Intuitively, such a triple states that given precondition \htPreConVar, command \cmdVar either 
    (i)~diverges, i.e., loops forever, 
    (ii)~abruptly terminates via \cmdExit or 
    (iii)~terminates normally and postcondition \htPostConVar holds afterwards.
    In the following we make this intuition precise such that we can use it to show the correctness of our verification approach.
    The annotated semantics act as an intermediary between high-level reasoning steps, e.g., using obligations to track which thread is going to \cmdExit, and the actual program executions.
    Hence, we use this connection to define the meaning of Hoare triples.

    Note that a reduction in the annotated semantics can get stuck in contrast to a reduction in the plain semantics.
    Therefore, a specification
    \hoareTriple{\htPreConVar}{\cmdVar}{\htPostConVar}
    additionally expresses that reduction of \cmdVar in the annotated semantics does not get stuck.

    \paragraph{Interpretations}
    We interpret Hoare triples in terms of a model relation \htModelsSymb and an auxiliary safety relation \isSafe{\rbVar}{\cmdVar}.
    Intuitively, a continuation \contVar is safe under a complete resource bundle \rbVar if \rbVar provides all necessary ghost resources such that the reduction of  $(\rbVar, \contVar)$ does not get stuck.
    We write \isResAnnot{\rtpVar}{\tpVar} to express that \rtpVar is an annotated version of \tpVar, containing the same threads but each equipped with a resource bundle.
    
        \begin{mydefinition}[Annotation of Thread Pools]
            We say that \rtpVar is an \emph{annotation} of \tpVar and write \isResAnnot{\rtpVar}{\tpVar} if $\domOf{\tpVar} = \domOf{\rtpVar}$ and if for every thread ID $\tidVar \in \domOf{\tpVar}$ there exists a resource bundle $\rbVar \in \ResourceBundleSet$ such that $\rtpVar(\tidVar) = (\rbVar,\, \tpVar(\tidVar))$.
        \end{mydefinition}
    
    \begin{restatable}[Safety]{mydefinition}{defSafety}\label{def:Safety}
        We define the safety predicate $\isSafeName \subseteq \ResourceBundleSet \times \ContSet$ coinductively as the greatest solution (with respect to $\subseteq$) of the following equation:
        $$
        \begin{array}{l}
                \isSafe{\rbVar}{\contVar}\ =\ \complete{\rbVar} \rightarrow\\
                \forall \tpVar, \nextTpVar.\
                \forall \tidVar \in \domOf{\tpVar}.\
                \forall \rtpVar.\\
                \ \
                        \tpVar(\tidVar) = \contVar \wedge
                        \tpRedStep{\tpVar}{\tidVar}{\nextTpVar} \wedge
                        \isResAnnot{\rtpVar}{\tpVar} \wedge
                        \rtpVar(\tidVar) = (\rbVar, \contVar)
                        \rightarrow\\
                \ \
                                \exists \grtpVar.\
                                \exists \nextRtpVar.\
                                        \grtpMultRedStep{\rtpVar}{\tidVar}{\grtpVar} \wedge
                                        \ngrtpRedStep{\grtpVar}{\tidVar}{\nextRtpVar} \wedge
                                        \isResAnnot{\nextRtpVar}{\nextTpVar}\\
                \ \ \phantom{\exists \grtpVar.\ }
                                        \wedge 
                                        \forall (\rbVar^*, \contVar^*) \in 
                                            \rangeOf{\nextRtpVar} \setminus \rangeOf{\rtpVar}.\
                                                    \isSafe{\rbVar^*}{\contVar^*}
        \end{array}
        $$
    \end{restatable}

    Intuitively, a Hoare triple 
    \hoareTriple{\htPreConVar}{\cmdVar}{\htPostConVar}
    holds in our model relation \htModelsSymb if the following two conditions are met:
    \begin{itemize}
            \item[(1)] Any resources \rbPreVar fulfilling precondition \htPreConVar suffice to reduce \cmdVar without getting stuck, i.e., 
            \isSafe{\rbPreVar}{\cmdVar; \contDone}.
            \item[(2)] If reduction of $(\rbPreVar, \cmdVar; \contDone)$ terminates in $(\rbPostVar, \contDone)$ (i.e., reduction does neither abruptly \cmdExit nor diverge), then \rbPostVar fulfils postcondition \htPostConVar.
    \end{itemize}
    Note that given (1), property (2) is equivalent to 
    \isSafe{\rbPostVar}{\cmdVar; \contVar}
    for any continuation \contVar safe under \rbPostVar.

    \begin{restatable}[Hoare Triple Model Relation]{mydefinition}{defHoareTripleModelRelation}\label{def:HoareTripleModelRelation}
        We define the Hoare triple model relation \htModelsSymb such that
        $$
        \begin{array}{c}
                \htModels{\htPreConVar}{\cmdVar}{\htPostConVar}\\
                \Longleftrightarrow\\
                \begin{array}{l}
                        \forall \frbVar.\
                        \forall \contVar.\
                                (\forall \rbPostVar.\
                                        \assModels{\rbPostVar}{\htPostConVar}
                                        \ \rightarrow\ 
                                        \isSafe{\rbPostVar \tupleCup \frbVar}{\contVar})
                        \\
                        \phantom{\forall \frbVar.\
                                                \forall \contVar.\ }
                                \rightarrow\ 
                                (\forall \rbPreVar.\
                                        \assModels{\rbPreVar}{\htPreConVar}
                                        \ \rightarrow \
                                        \isSafe{\rbPreVar \tupleCup \frbVar}{\,\cmdVar\, ; \contVar})
                \end{array}
        \end{array}
        $$
    \end{restatable}
    
    Note that compliance with the frame rule directly follows from above definition, i.e., 
    \htModels{\htPreConVar}{\cmdVar}{\htPostConVar}
    implies
    \htModels
            {\htPreConVar \slStar \htFrameVar}
            {\cmdVar}
            {\htPostConVar \slStar \htFrameVar}
    for any frame $\htFrameVar \in \AssertionSet$.
    Further, every specification \hoareTriple{\htPreConVar}{\cmdVar}{\htPostConVar} we can derive with our proof rules also holds in our model.
    
    \begin{restatable}[Soundness of Hoare Triples]{mylemma}{lemSoundnessHoareTriples}\label{lem:SoundnessHoareTriples}
        Let \htProves{\htPreConVar}{\cmdVar}{\htPostConVar}.
        Then \htModels{\htPreConVar}{\cmdVar}{\htPostConVar} holds.
    \end{restatable}
    \begin{restatable}{proof}{proofLemSoundnessHoareTriples}
        By induction on the derivation of \htProves{\htPreConVar}{\cmdVar}{\htPostConVar}.
    \end{restatable}

    This lemma bridges the gap between our verification approach and the annotated semantics.
    That is, whenever we can prove a specification
    \htProves{\htPreConVar}{\cmdVar}{\htPostConVar},
    we know that command \cmdVar can be safely reduced in the annotated semantics given precondition \htPreConVar.
    If this reduction terminates normally (i.e. in case it neither diverges nor abruptly terminates), postcondition \htPostConVar holds in the final state.

\subsection{Soundness Proof}\label{subsec:SoundnessProof}
    In the proof of Soundness Theorem~\ref{theo:Soundness}, we show that programs which provably discharge all their \cmdExit obligations and start without credits terminate under fair scheduling.
    That is, we show that such programs cannot have a corresponding fair reduction sequence.
    To be able to refer to the eventually discharged obligations, the proof requires us to refer to annotated reduction sequences.
    The following lemma allows us to construct such an annotated reduction sequence from an unannotated one.

    \begin{myobservation}[\contDone Safe]\label{observation:DoneSafe}
        \isSafe{\rbVar}{\contDone}
        holds for all \completeTerm \rbVar.
    \end{myobservation}

    \begin{restatable}{mylemma}{lemModelAllowsAnnotation}\label{lem:ModelRelationAllowsAnnotation}
        Let \htModels{\htPreConVar}{\cmdVar}{\htPostConVar}, \assModels{\rbPreVar}{\htPreConVar} and \complete{\rbPreVar}.
        Furthermore, let \tpRSeq{\N} be fair with
        $\tpVar_0 = \setOf{(\tidVar, \cmdVar;\contDone)}$.
        There exists a fair annotated reduction sequence \rtpRSeq{\N} with
        $\rtpVar_0 = \setOf{(\tidVar, (\rbPreVar, \cmdVar;\contDone))}$.
    \end{restatable}
    
    \begin{proof}
        According to the definition of \htModelsSymb (cf. Definition~\ref{def:HoareTripleModelRelation}) the following implication holds:
        $$
        \begin{array}{l}
            \forall \frbVar.\
            \forall \contVar.\
                    (\forall \rbPostVar.\
                        \assModels{\rbPostVar}{\htPostConVar}
                        \ \rightarrow\ 
                        \isSafe{\rbPostVar \tupleCup \frbVar}{\contVar})
                    \\
                    \phantom{\forall \frbVar.\
                                        \forall \contVar.\ }
                    \rightarrow\ 
                    (\forall \rbPreVar^*.\
                            \assModels{\rbPreVar^*}{\htPreConVar}
                            \ \rightarrow \
                            \isSafe{\rbPreVar^* \tupleCup \frbVar}{\,\cmdVar\, ; \contVar})
        \end{array}
        $$
        We can instantiate this to
        $$
        \begin{array}{l}
            (\forall \rbPostVar.\
             \assModels{\rbPostVar}{\htPostConVar}
             \ \rightarrow\ 
             \isSafe{\rbPostVar}{\contDone})
             \\
             \rightarrow\ 
             (\forall \rbPreVar^*.\
                    \assModels{\rbPreVar^*}{\htPreConVar}
                    \ \rightarrow \
                    \isSafe{\rbPreVar^*}{\,\cmdVar\, ; \contDone}).
          \end{array}
        $$
        According to Observation~\ref{observation:DoneSafe}, \isSafe{\rbPostVar}{\contDone} holds for any \completeTerm \rbPostVar.
        Hence, the implication's precondition holds trivially and we get
        $$
        \forall \rbPreVar^*.\
                \assModels{\rbPreVar^*}{\htPreConVar}
                \ \rightarrow \
                \isSafe{\rbPreVar^*}{\,\cmdVar\, ; \contDone}
        $$
        and in particular \isSafe{\rbPreVar}{\cmdVar; \contDone}.
        
        In the following, we construct the witness \rtpRSeq{\N} inductively from the unannotated reduction sequence \tpRSeq{\N}.
        In the lemma, we already defined the reduction sequence's start
        $\rtpVar_0 = \setOf{(\tidVar,\, (\rbPreVar, \cmdVar;\contDone))}$.
        
        Assume we got an annotation \rtpRSeq{\setOf{0,...,m}} of the prefix \tpRSeq{\setOf{0,...,n}} where \isSafe{\rbVar^*}{\contVar^*} holds for every
        $(\rbVar^*, \contVar^*) \in \rangeOf{\rtpVar_m}$.
        Note that $m \geq n$, since the annotated prefix might contain ghost steps.
        
        Let $\tidVar_n$ be the thread ID corresponding to the continuation reduced in reduction step $n$, i.e., \tpRedStep{\tpVar_n}{\tidVar_n}{\tpVar_{n+1}}.
        There exist $\rbVar_{m}$ and $\contVar_{m}$ with $\rtpVar_{m}(\tidVar_n) = (\rbVar_{m}, \contVar_{m})$.
        By the assumption \isSafe{\rbVar_m}{\contVar_m} we get the existence of an annotated thread pool $\nextRtpVar$ with \isResAnnot{\nextRtpVar}{\tpVar_{n+1}}, $\nextRtpVar(\tidVar_n) = (\nextRbVar, \nextContVar)$ and also \isSafe{\nextRbVar}{\nextContVar} for some $\nextRbVar, \nextContVar$.
                
        According to \isSafe{\rbVar_m}{\contVar_m} there exists \grtpVar with
        \grtpMultRedStep{\rtpVar}{\tidVar}{\grtpVar}
        and
        \ngrtpRedStep{\grtpVar}{\tidVar}{\nextRtpVar}.
        That is, we get the existence of a (potentially empty) sequence of thread pools
        $\rtpVar_{m+1},...,\rtpVar_{m+h}$ corresponding to \grtpMultRedStep{\rtpVar}{\tidVar}{\grtpVar}.
        By setting $\rtpVar_{m+h+1} := \nextRtpVar$ we obtain an annotation \rtpRSeq{\setOf{0,...,m+h+1}} of the extended prefix \tpRSeq{\setOf{0,...,n+1}}.
        We get \isSafe{\rbVar^*}{\contVar^*} for all 
        $(\rbVar^*, \contVar^*) \in \rangeOf{\rtpVar_{m+h+1}}$.

        By induction we get the claimed annotated reduction sequence \rtpRSeq{\N} with
        $\rtpVar_0 = \setOf{(\tidVar,\, (\rbPreVar, \cmdVar;\contDone))}$.
        By construction of \rtpRSeq{\N}, there exists an annotated reduction step \ngrtpRedStep{\rtpVar_j}{\tidVar_i}{\rtpVar_{j+1}} for every step \tpRedStep{\tpVar_i}{\tidVar_i}{\tpVar_{i+1}} in the plain sequence.
        Hence, the construction preserves the fairness of \tpRSeq{\N}.
    \end{proof}

    As next step, we show that an annotated reduction sequence such as \rtpRSeq{\N} constructed above, cannot start with initial resources of the form 
    $\rbPreVar = (\multiset{\obVar_0}, 0)$.
    We do this by analysing the program order graph 
    $\progOrdGraph(\rtpRSeq{\N})$
    defined in the following.
    In this graph, every node $i$ represents the $i^\text{th}$ reduction step of the sequence, i.e.,
    \rtpRedStep{\rtpVar_i}{\tidVar_i}{\rtpVar_{i+1}}.
    Edges have the form 
    $(i, \tidVar_j, n, j)$.
    Such an edge expresses that either
    (i)~\rtpRedStep{\rtpVar_j}{\tidVar_j}{\rtpVar_{j+1}} is the first step of a thread forked in step $i$
    or
    (ii)~$j$ is the next index representing a reduction of thread $\tidVar_i$ (in which case $\tidVar_i = \tidVar_j$ holds).
    In both cases, $n$ represents the name of the reduction rule applied in step
    \rtpRedStep{\rtpVar_j}{\tidVar_j}{\rtpVar_{j+1}}.

    \begin{mydefinition}[Program Order Graph]\noindent\\
        Let $\rtpVar_0 = \rtp{\tidVar_0}{\rbVar_0}{\contVar_0}$ be an annotated thread pool and \rtpRSeq{I} an annotated reduction sequence.
        Furthermore, let \RRedRuleNameSet be the set of names referring to reduction rules defining the relations \ngrtpRedStepSymb, \grtpRedStepSymb and \rstRedStepSymb.
        
        Below, we define the \emph{program order graph} $\progOrdGraph(\rtpRSeq{\N}) = (\N, E)$ for $\rtpRSeq{\N}$ where $E \subseteq \N \times \N \times \RRedRuleNameSet \times \N$.
        We define the set of edges $E$ as the smallest set meeting the following requirements:
        
        Let $i, j \in \N$ be indices denoting reduction steps \rtpRedStep{\rtpVar_i}{\tidVar_i}{\rtpVar_{i+1}} and \rtpRedStep{\rtpVar_j}{\tidVar_j}{\rtpVar_{j+1}} for some thread IDs $\tidVar_i, \tidVar_j \in \N$.
        Let $n \in \RRedRuleNameSet$ be the name of the reduction rule applied during step \rtpRedStep{\rtpVar_i}{\tidVar_i}{\rtpVar_{i+1}} in the following sense:
        \begin{itemize}
            \item In case the step is an application of reduction rule \RARedTPLiftName in combination with some single thread reduction rule $n_{\fixedPredNameFont{st}}$, then $n = n_{\fixedPredNameFont{st}}$.
            \item Otherwise, $n$ denotes the applied (real or ghost) thread pool reduction rule.
        \end{itemize}
    
        \noindent
        With these choices, $(i, \tidVar_j, n, j) \in E$ holds if one of the following holds:
        \begin{itemize}
            \item $\tidVar_i = \tidVar_j$ and 
                    $j = \min(\{  l\in\N \ | \ l > i\ \wedge\ \rtpRedStep{\rtpVar_l}{\tidVar_j}{\rtpVar_{l+1}}  \})$
            \item $\tidVar_i \neq \tidVar_j$,\
                    $\domOf{\rtpVar_{i+1}} \setminus \domOf{\rtpVar_i} = \setOf{\tidVar_j}$
                    and\\
                    $j = \min(\{ l \in \N\ |\ l > i\ \wedge\ \rtpRedStep{\rtpVar_l}{\tidVar_j}{\rtpVar_{l+1}}\})$,
        \end{itemize}
        
        \noindent
        We define node 0 to be the root of the program order graph. In case the choice of reduction sequence is clear from the context, we write \progOrdGraph instead of $\progOrdGraph(\rtpRSeq{\N})$.
    \end{mydefinition}
    
\NewDocumentCommand{\exProg}{}
        {\ensuremath{\cmdVar_\text{ex2}}\xspace}

    Consider the program 
    $$
        \exProg = 
                \cmdFork
                [
                        (\cmdFork[\cmdLoop[\cmdSkip]];
                        \cmdExit)
                ];
                \cmdLoop[\cmdSkip]
    $$
    It forks a thread and busy-waits for it to \cmdExit the program.
    In turn, the forked thread forks another thread that busy-waits as well, before it abruptly terminates the program.
    Figure~\ref{fig:ExampleReductionSequence} presents one possible example reduction sequence and its program order graph.
    
    \begin{figure}
            $$
            \hspace{-4cm}
            \begin{array}{l | l  l | l}
                    \text{ID}
                            &\text{Thread pool}
                            &
                            &\text{Applied reduction rule}\\
                    0
                            &\{ (0, (\rtpEmpty, \exProg)) \} 
                            &\grtpRedStepSymb[0]
                            &\proofRuleHint{\GSObCredIntroName}\\
                    1
                            &\{ (0, ((\multiset{1}, 1), \exProg)) \}
                            &\ngrtpRedStepSymb[0]
                            &\proofRuleHint{\RARedTPLiftName} +
                                    \proofRuleHint{\RARedSTForkName}\\
                    2
                            &\{ (0, ((\multiset{0}, 1), \cmdLoop[\cmdSkip])),\ \
                                        (1, ((\multiset{1}, 0), 
                                                    \cmdFork[\cmdLoop[\cmdSkip]]; \cmdExit))
                                 \}
                            &\grtpRedStepSymb[1]
                            &\proofRuleHint{\GSObCredIntroName}\\
                    3
                            &\{ (0, ((\multiset{0}, 1), \cmdLoop[\cmdSkip])),\ \
                                        (1, ((\multiset{2}, 1), 
                                                    \cmdFork[\cmdLoop[\cmdSkip]]; \cmdExit))
                                 \}
                            &\ngrtpRedStepSymb[1]
                            &\proofRuleHint{\RARedTPLiftName} +
                                    \proofRuleHint{\RARedSTForkName}\\        
                    4
                            &\{ (0, ((\multiset{0}, 1), \cmdLoop[\cmdSkip])),\ \
                                        (1, ((\multiset{2}, 0), \cmdExit)),\ \
                                        (2, ((\multiset{0}, 1), \cmdLoop[\cmdSkip]))
                                 \}
                            &\ngrtpRedStepSymb[2]
                            &\proofRuleHint{\RARedTPLiftName} +
                                    \proofRuleHint{\RARedSTLoopName}\\    
                    5
                            &\{ (0, ((\multiset{0}, 1), \cmdLoop[\cmdSkip])),\ \
                                        (1, ((\multiset{2}, 0), \cmdExit)),\ \
                                        (2, ((\multiset{0}, 1), \cmdLoop[\cmdSkip]))
                                 \}
                            &\ngrtpRedStepSymb[0]
                            &\proofRuleHint{\RARedTPLiftName} +
                                    \proofRuleHint{\RARedSTLoopName}\\    
                    6
                            &\{ (0, ((\multiset{0}, 1), \cmdLoop[\cmdSkip])),\ \
                                        (1, ((\multiset{2}, 0), \cmdExit)),\ \
                                        (2, ((\multiset{0}, 1), \cmdLoop[\cmdSkip]))
                                 \}
                            &\ngrtpRedStepSymb[0]
                            &\proofRuleHint{\RARedTPLiftName} +
                                    \proofRuleHint{\RARedSTLoopName}\\    
                    7
                            & \quad\quad\quad\quad\quad \quad\quad\quad\quad\quad
                                     \quad\quad\quad\quad\quad
                                     \vdots
                            &\ngrtpRedStepSymb[0]
                            &\proofRuleHint{\RARedTPLiftName} +
                                    \proofRuleHint{\RARedSTLoopName}\\    
                    \vdots
                            &  \quad\quad\quad\quad\quad \quad\quad\quad\quad\quad
                                    \quad\quad\quad\quad\quad
                                    \vdots
                            &
                            &\phantom{\proofRuleHint{\RARedTPLiftName}}\ \vdots
            \end{array}
            $$

            \begin{centering}
                    \includegraphics[scale=1.20]{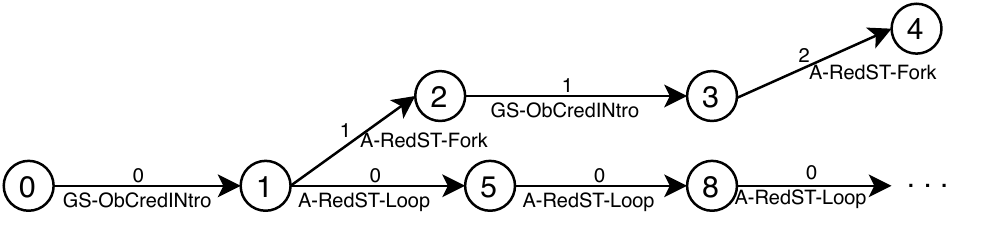}
            \end{centering}
    
            \caption{Possible reduction sequence and program order graph of     
                    $
                    \exProg = 
                            \cmdFork
                            [
                                    (\cmdFork[\cmdLoop[\cmdSkip]];
                                    \cmdExit)
                            ];
                            \cmdLoop[\cmdSkip]
                    $.
                    Threads with ID \tidVar, resource bundle \rbVar and continuation \contVar are depicted as a tuple $(\tidVar, (\rbVar, \contVar))$.
                    }
            \label{fig:ExampleReductionSequence}
    \end{figure}

    For any program order graph \progOrdGraph, we call a subgraph $\progOrdGraph_s$ of \progOrdGraph sibling-closed if, for each node $n$ of the subgraph, $n$'s predecessors' successors are also in the subgraph. 
    In other words, for each fork step node, either both the next step of the forking thread and the first step of the forked thread are in the subgraph, or neither are.
    
    Sibling-closed prefixes of program order graphs have the special property that the sum of the obligations held by the threads reduced in their leaves equals the sum of the credits held by these threads.
    This property forms a crucial part of our soundness argument, as we see in the following.
    
    \begin{mydefinition}[Sibling-Closed Subgraph]
            Let $\progOrdGraph$ be a program order graph and $\progOrdGraph_s$ a subgraph of \progOrdGraph.
            For a node $n \in \progOrdGraph_s \setminus \setOf{0}$, let 
            $p_n$ be the predecessor of $n$ in \progOrdGraph, i.e., the node for which $\theta, m$ exist such that
            $(p_n, \theta, m, n) \in \edgesOf{\progOrdGraph}$.
            Further, let
            $S_n = 
                \setOf{j  \ \ | \ \ 
                    \exists \tidVar. \exists m.\ (p_n, \tidVar, m, j) \in \edgesOf{\progOrdGraph}}
            $
            be the set containing $n$ and $n$'s siblings in \progOrdGraph.

            We call $\progOrdGraph_s$ \emph{sibling-closed} if for all 
            $n \in \nodesOf{\progOrdGraph_s}$
            and all
            $s \in S_n$, it holds that
            $s \in \progOrdGraph_s$.
    \end{mydefinition}

    \begin{restatable}{mylemma}{lemLeafObsEqLeafCreds}\label{lem:LeafObsEqLeafCreds}
            Let $\progOrdGraph_p$ be a finite 
            sibling-closed
            prefix of some program order graph $\progOrdGraph(\rtpRSeq{\N})$
            with $\rtpVar_0 = \setOf{(\tidVar_0, (\rbVar_0, \contVar_0))}$
            for some \completeTerm $\rbVar_0$.
            Let $L$ be set of leaves of $\progOrdGraph_p$.
            For all $l \in L$ choose $\tidVar_l, \obVar_l, \credVar_l, \contVar_l$ such that \rtpRedStep{\rtpVar_l}{\tidVar_l}{\rtpVar_{l+1}} and 
            $\rtpVar_l(\tidVar_l) = ((\multiset{\obVar_l}, \credVar_l), \contVar_l)$.

            The sum of the numbers of obligations held by the threads reduced in the leaves of $\progOrdGraph_p$ equals the sum of the numbers of credits held by these threads, i.e,
            $$
                    \Sigma_{l \in L}\, \obVar_l
                    \ \ = \ \
                    \Sigma_{l \in L}\, \credVar_l
            $$
    \end{restatable}

    \begin{restatable}{proof}{proofLemLeafObsEqLeafCreds}
            Proof by induction on the size of $\progOrdGraph_p$.
    \end{restatable}

    Remember that reduction sequences are infinite by definition and consider a fair annotated reduction sequence \rtpRSeq{\N}.
    The only construct in our language that introduces non-termination are loops.
    Since \rtpRSeq{\N} is fair and does not get stuck, it must contain loop steps, i.e., reduction steps that result from an application of reduction rule \RARedTPLiftName in combination with \RARedSTLoopName.
    The latter rule, however, requires a credit but forbids the looping thread to hold an obligation.
    Obligations and credits can only be spawned simultaneously.
    Hence, another thread must hold the obligation and will eventually \cmdExit.
    
    Intuitively, in the annotated semantics, reduction under fair scheduling must either get stuck or terminate.
    The following lemma makes this intuition precise.

    \begin{restatable}{mylemma}{lemNoFairAnnotatedRedSeq}\label{lem:NoFairAnnotatedReductionSequence}
            There are no fair annotated reduction sequences \rtpRSeq{\N} with
            $\rtpVar_0 = \setOf{(\tidVar_0, ((\multiset{\obVar_0}, 0), \contVar))}$
            for any $\tidVar_0, \obVar_0, \contVar$.
    \end{restatable}
    \begin{proof}
    {
        \NewDocumentCommand{\varSpecFont}{m}{\fixedPredNameFont{#1}}
        
        \NewDocumentCommand{\nodeVar}{}{\ensuremath{v}\xspace}
        
        \NewDocumentCommand{\loopNode}{}
                {\ensuremath{\nodeVar_\varSpecFont{l}}\xspace}
        \NewDocumentCommand{\loopTid}{}
                {\ensuremath{\tidVar_\varSpecFont{l}}\xspace}
        \NewDocumentCommand{\loopRb}{}
                {\ensuremath{\rbVar_\varSpecFont{l}}\xspace}
        \NewDocumentCommand{\loopCont}{}
                {\ensuremath{\contVar_\varSpecFont{l}}\xspace}
        
        \NewDocumentCommand{\spawnNode}{}
                {\ensuremath{\nodeVar_\varSpecFont{s}}\xspace}
        \NewDocumentCommand{\spawnTid}{}
                {\ensuremath{\tidVar_\varSpecFont{s}}\xspace}
        
        \NewDocumentCommand{\exitNode}{}
                {\ensuremath{\nodeVar_\varSpecFont{e}}\xspace}
        \NewDocumentCommand{\exitTid}{}
                {\ensuremath{\tidVar_\varSpecFont{e}}\xspace}
        \NewDocumentCommand{\exitRb}{}
                {\ensuremath{\rbVar_\varSpecFont{e}}\xspace}
        \NewDocumentCommand{\exitCont}{}
                {\ensuremath{\contVar_\varSpecFont{e}}\xspace}

        Assume such a reduction sequence exists.       
        Consider the program order graph \progOrdGraph of $(\rtpVar_i)_{i\in\N}$. 
        Since, $(\rtpVar_i)_{i\in\N}$ is infinite, it contains loop-steps, i.e., reduction steps $\rtpRedStep{\rtpVar_i}{\tidVar}{\rtpVar_{i+1}}$ resulting from applications of \RARedTPLiftName in combination with \RARedSTLoopName. 
        Accordingly, \progOrdGraph is infinite, too, and contains loop-edges of the form $(i, \tidVar, \RARedSTLoopName, j)$.
        
        Let \progOrdGraphLF be its maximal loop-edge-free sibling-closed prefix.
        By analyzing the leaves of \progOrdGraphLF, we prove that \progOrdGraph and \rtpRSeq{\N} are finite, which contradicts our original assumption about the existence of a fair reduction sequence \tpRSeq{\N}.
        
        Let $\loopNode\in \nodesOf{\progOrdGraph}$ be a node of the full graph for which an edge of the form $(\loopNode, \loopTid, \RARedSTLoopName, j) \in \edgesOf{\progOrdGraph}$ exists and for which the length of the path from root 0 to \loopNode is minimal. 
        That is, said path does not contain any loop-edges. Therefore, \loopNode is a leaf of the prefix graph \progOrdGraphLF and \rtpRedStep{\rtpVar_{\loopNode}}{\loopTid}{\rtpVar_{\loopNode+1}} is a loop-step.
        
        According to reduction rule \RARedSTLoopName, there exist $\loopRb, \loopCont$ with $\rtpVar_{\loopNode}(\loopTid) = (\loopRb, \loopCont)$ and \assModels{\loopRb}{\credit}. 
        Additionally, the rule also requires resource bundle \loopRb to contain no obligations.
        According to Lemma~\ref{lem:LeafObsEqLeafCreds}, there exists a leaf
        $\exitNode \in \leavesOf{\progOrdGraphLF} \setminus \setOf{\loopNode}$,
        a thread ID 
        $\exitTid\in \N\setminus \setOf{\loopTid}$,
        a continuation
        $\exitCont\in \ContSet$ 
        and a resource bundle $\exitRb\in \ResourceBundleSet$
        containing the obligation such that $\rtpVar_{\exitNode}(\exitTid) = (\exitRb, \exitCont)$.
  
        Generally, there are three possible reasons why some node $\nodeVar \in \nodesOf{\progOrdGraph}$ is a leaf of \progOrdGraphLF, all of which depend on the associated reduction step \rtpRedStep{\rtpVar_{\nodeVar}}{\tidVar_\nodeVar}{\rtpVar_{\nodeVar+1}}:
        (i)~Because it is a loop-step, 
        (ii)~because it is a thread termination step, i.e.,
            $\rtpVar_{\nodeVar}(\tidVar_\nodeVar) = (\rbVar_\nodeVar, \contDone)$
            and
            $\rtpVar_{\nodeVar+1} = 
                \rtpVar_{\nodeVar} \rtpRem \tidVar_\nodeVar$
        or
        (iii)~because it is an abrupt exit step.

        However, (i)~\rtpRedStep{\rtpVar_{\exitNode}}{\exitTid}{\rtpVar_{\exitNode+1}} cannot be a loop-step according to \RARedSTLoopName, since it prohibits \exitRb to hold any obligations.
        Similarly, (ii)~\rtpRedStep{\rtpVar_{\exitNode}}{\exitTid}{\rtpVar_{\exitNode+1}} cannot be a thread termination step either, since reduction rule \RARedTPThreadTermName also  prohibits \exitRb to hold any obligations.
        
        Therefore, the only possibility is (iii), i.e., \rtpRedStep{\rtpVar_{\exitNode}}{\exitTid}{\rtpVar_{\exitNode+1}} abruptly exits the program and $\rtpVar_{\exitNode+1} = \rtpEmpty$. 
        Thereby, the program order graph \progOrdGraph and the annotated reduction sequence \rtpRSeq{\N} must be finite.
    }
    \end{proof}

    The combination of Lemmas~\ref{lem:SoundnessHoareTriples},
    \ref{lem:ModelRelationAllowsAnnotation} and~\ref{lem:NoFairAnnotatedReductionSequence} 
    allow a straight-forward proof of the soundness theorem.

    \SoundnessTheorem*
    \begin{restatable}{proof}{proofSoundnessTheorem}
        Assume there exists a fair reduction sequence $(\tpVar_i)_{i\in \N}$ starting with $\tpVar_0 = \setOf{(\tidVar_0, \cmdVar;\contDone)}$.
        By application of the Soundness Lemma for Hoare triples, Lemma~\ref{lem:SoundnessHoareTriples}, we get \htModels{\obs{n}}{\cmdVar}{\noObs}. 
        By Lemma~\ref{lem:ModelRelationAllowsAnnotation}, there exists a fair annotated reduction sequence \rtpRSeq{\N} starting with 
        $\rtpVar_0 = \setOf{(\tidVar_0, ((\multiset{n}, 0), \cmdVar; \contDone))}$
        
        However, by Lemma~\ref{lem:NoFairAnnotatedReductionSequence} the annotated reduction sequence \rtpRSeq{\N} does not exist and consequently neither does \tpRSeq{\N}.
    \end{restatable}

\NewDocumentCommand{\eventThread}{o}
        {\ensuremath{
                t_e
                \IfValueT{#1}{^{#1}}
        }\xspace}
\NewDocumentCommand{\waitThread}{o}
        {\ensuremath{
                t_w
                \IfValueT{#1}{^{#1}}
        }\xspace}        
    
\section{Future Work}\label{sec:FutureWork}
    We are currently formalizing the presented approach and its soundness proof in Coq.

    \paragraph{Ghost Signals}
    We plan to extend the verification approach described in this paper to a verification technique we call \emph{ghost signals}, which allows us to verify termination of busy-waiting for arbitrary events.
    Ghost signals come with an obligation to set the signal and a credit that can be used to busy-wait for the signal to be set.
    Consider a program with two threads: \eventThread eventually performs some event $X$  (such as setting a flag in shared memory) and \waitThread  busy-waits for $X$.
    By letting \eventThread  set the signal when it performs $X$, and thereby linking the ghost signal to the event, we can justify termination of \waitThread's busy-waiting.

    \paragraph{I/O Liveness as Abrupt Program Exit}
    In concurrent work we encode I/O liveness properties as abrupt program termination following a conjecture of \citet{Jacobs2018ModularTerminationVerification}.
    Consider a \emph{non-terminating} server \serverVar which shall reply to all requests.
    We can prove liveness of \serverVar using the following methodology:
    \begin{itemize}
        \item For some arbitrary, but fixed $N$, assume that responding to the $N^{\text{th}}$ request abruptly terminates the whole program.
        \item Prove that the program always terminates.
    \end{itemize}

    \begin{figure}
          \begin{center}
                \includegraphics[scale=0.86]{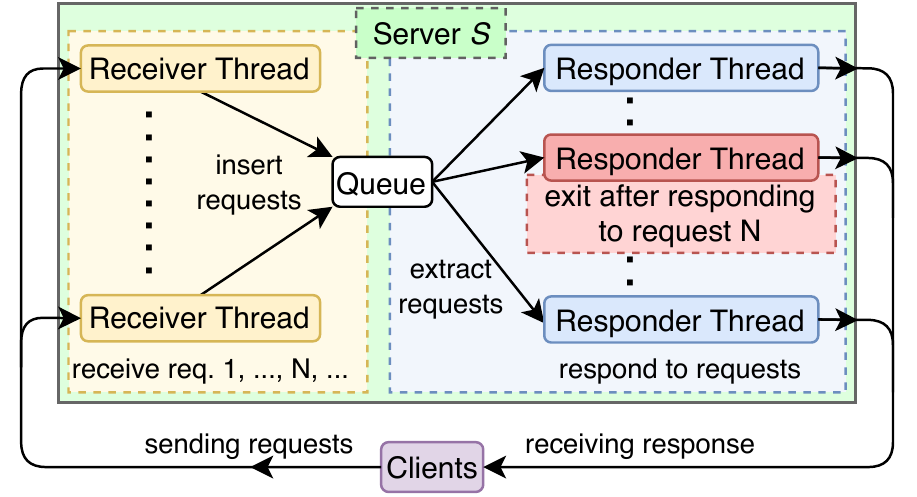}
            \end{center}
        
        \vspace{-0.3cm}
        \caption
            {Server \serverVar receiving and replying to requests.
             Threads communicating via shared queue.}
        \label{fig:ServerMultipleThreads}
    \end{figure}

    One can combine this approach with the one of the present paper to verify liveness of a server where multiple threads independently receive and handle requests. 
    Using a \emph{prophecy variable}~\cite{Jung2020POPLProphecyVariables}, one can determine ahead of time which thread will receive the exiting request. 
    The other threads can then be seen as busy-waiting for this thread to exit.

\paragraph{Combination}

    We plan to combine the two approaches sketched above and conjecture that the combination will be expressive enough to verify liveness of programs such as the server \serverVar presented in Figure~\ref{fig:ServerMultipleThreads}.
    It runs a set of receiving and a set of responding threads communicating via a shared queue.
    The responding threads busy-wait for requests to arrive in the queue.
    In order to verify liveness of \serverVar, we need to show that some thread
    eventually abruptly terminates the program by acquiring and responding to the $N^\text{th}$ request.
    This requires us to prove that the responding threads' busy-waiting for requests to arrive in the shared queue terminates.
    
    In order to demonstrate the approach's usability, we plan to implement it in VeriFast~\cite{Jacobs2011Verifast} and prove liveness of \serverVar.

\section{Related Work}\label{sec:RelatedWork}
\NewDocumentCommand{\TadaLive}{}{TaDA Live\xspace}
\NewDocumentCommand{\Lili}{}{LiLi\xspace}

    \citet{Liang2016LiliAPL, Liang2017LiliProgressOC} propose \Lili, a separation logic to verify liveness of blocking constructs implemented via busy-waiting.
    In contrast to our verification approach, theirs is based on the idea of contextual refinement.
    \Lili does not support forking nor structured parallel composition.

    \citet{DOsualdo2019arxivTaDALive} propose \TadaLive, a separation logic for verifying termination of busy-waiting.
    This logic allows to modularly reason about fine-grained concurrent programs and blocking operations that are implemented in terms of busy-waiting and non-blocking primitives.
    It uses the concept of obligations to express thread-local liveness invariants, e.g., that a thread eventually releases an acquired lock.
    \TadaLive is expressive enough to verify CLH and spin locks.
    The current version supports structured parallel composition instead of unstructured forking.
    Comparing their proof rules to ours, it is fair to say that our logic is simpler. 
    Of course, theirs is much more powerful. 
    We hope to be able to extend ours as sketched above while remaining simpler.

\section{Conclusion}\label{sec:Conclusion}
    In this paper we proposed a separation logic to verify the termination of programs where some threads abruptly terminate the program and others busy-wait for abrupt termination.
    We proved our logic sound and illustrated its application.

    Abrupt termination can be understood as an approximation of a general event.
    We outlined our vision on how to extend our approach to verify the termination of busy-waiting for arbitrary events.
    We have good hope that the final logic will be as expressive as \TadaLive proposed by \citet{DOsualdo2019arxivTaDALive} while remaining conceptually simpler.

    Further, we sketched our vision to combine this extended work with concurrent work on the encoding of I/O liveness properties as abrupt termination.
    We illustrated that this combination will be expressive enough to verify liveness of concurrent programs where multiple threads share I/O responsibilities.

\bibliographystyle{plainnat}
\bibliography{bibliography}

\appendix

\end{document}